\documentclass[11pt]{article}

\usepackage[T1]{fontenc}
\usepackage[utf8]{inputenc}
\usepackage[margin=1in]{geometry}
\usepackage{amsmath,mathtools,amsthm}
\usepackage{newtxtext,newtxmath}
\usepackage{microtype}
\allowdisplaybreaks
\usepackage{booktabs,tabularx,array}
\usepackage{authblk}
\usepackage{enumitem}
\usepackage{xcolor}
\usepackage{tikz}
\usetikzlibrary{arrows.meta,positioning,calc,fit,backgrounds,decorations.pathreplacing,shapes.geometric}
\usepackage{caption}
\usepackage{float}
\usepackage[numbers,square]{natbib}
\let\cite\citep

\usepackage[hidelinks,unicode]{hyperref}
\usepackage[nameinlink,noabbrev,capitalise]{cleveref}

\setlist[itemize]{leftmargin=1.55em,itemsep=1.5pt,topsep=3pt}
\setlist[enumerate]{leftmargin=1.8em,itemsep=1.5pt,topsep=3pt}

\newtheorem{theorem}{Theorem}[section]
\newtheorem{lemma}[theorem]{Lemma}
\newtheorem{proposition}[theorem]{Proposition}
\newtheorem{corollary}[theorem]{Corollary}

\theoremstyle{definition}
\newtheorem{definition}[theorem]{Definition}
\theoremstyle{remark}
\newtheorem{remark}[theorem]{Remark}

\crefname{theorem}{Theorem}{Theorems}
\crefname{lemma}{Lemma}{Lemmas}
\crefname{proposition}{Proposition}{Propositions}
\crefname{corollary}{Corollary}{Corollaries}
\crefname{definition}{Definition}{Definitions}
\crefname{observation}{Observation}{Observations}
\crefname{remark}{Remark}{Remarks}
\crefname{section}{Section}{Sections}
\crefname{figure}{Figure}{Figures}
\crefname{table}{Table}{Tables}
\crefname{appendix}{Appendix}{Appendices}

\AddToHook{env/theorem/begin}{\crefalias{theorem}{theorem}}
\AddToHook{env/lemma/begin}{\crefalias{theorem}{lemma}}
\AddToHook{env/proposition/begin}{\crefalias{theorem}{proposition}}
\AddToHook{env/corollary/begin}{\crefalias{theorem}{corollary}}
\AddToHook{env/definition/begin}{\crefalias{theorem}{definition}}
\AddToHook{env/observation/begin}{\crefalias{theorem}{observation}}
\AddToHook{env/remark/begin}{\crefalias{theorem}{remark}}

\definecolor{figblue}{HTML}{0072B2}
\definecolor{figsky}{HTML}{56B4E9}
\definecolor{figgreen}{HTML}{009E73}
\definecolor{figorange}{HTML}{E69F00}
\definecolor{figvermillion}{HTML}{D55E00}
\definecolor{figpurple}{HTML}{CC79A7}
\definecolor{figslate}{HTML}{4D4D4D}
\definecolor{figlightblue}{HTML}{EAF4FB}
\definecolor{figlightgreen}{HTML}{EAF7F2}
\definecolor{figlightorange}{HTML}{FFF4DF}
\definecolor{figlightpurple}{HTML}{F8EDF5}

\newcommand{\ALG}{\operatorname{ALG}}

\newcommand{\one}{\mathbf 1}

\newcommand{\LPex}{\mathrm{ex}_{\mathrm{LP}}}
\newcommand{\vex}[1]{\mathrm{ex}_{#1}}
\newcommand{\reqdeg}[1]{d^{\mathrm{req}}_{#1}}
\newcommand{\Uplinks}{\mathcal U}
\newcommand{\Packing}{\mathcal P}
\newcommand{\PackM}{\mathcal P_M}
\newcommand{\PackH}{\mathcal P_H}
\newcommand{\npacked}{p_M}
\newcommand{\Missing}{m_{\mathrm{miss}}}
\newcommand{\HeavyTree}{T_H}
\newcommand{\TreeMass}{\mu_T}
\newcommand{\PackedMass}{\mu_P}
\newcommand{\OutsideMass}{\mu_{\mathrm{out}}}
\newcommand{\PackSlack}{s_{\mathrm{pack}}}
\newcommand{\localPackSlack}[1]{s^{\mathrm{pack}}_{#1}}
\newcommand{\Saving}{\mathsf{sav}}
\newcommand{\Skeleton}{I}
\newcommand{\FracHeavy}{H_{\mathrm{frac}}}
\newcommand{\ncomp}{c_I}
\newcommand{\nbridges}{b_I}
\newcommand{\SelfMass}{S_{\mathrm{self}}}
\newcommand{\localSelfMass}[1]{s^{\mathrm{self}}_{#1}}
\newcommand{\GateCopies}{\mathcal C_{\mathrm{gate}}}
\newcommand{\GateResidual}{R_{\mathrm{gate}}}
\newcommand{\GateOver}{O_{\mathrm{gate}}}

\tikzset{
  vnode/.style={circle,draw=figslate,fill=white,minimum size=5.2mm,inner sep=0pt},
  lightedge/.style={draw=figgreen,double=white,double distance=1.0pt,line width=1.8pt},
  heavyedge/.style={draw=figslate,line width=1.15pt},
  backedge/.style={draw=figvermillion,dashed,line width=1.1pt,-{Latex[length=2mm]}},
  packingedge/.style={draw=figblue,line width=2.6pt},
  gateedge/.style={draw=figorange,densely dotted,line width=2.25pt},
  flowbox/.style={draw=figblue,rounded corners=1mm,fill=figlightblue,align=center,minimum height=9mm,minimum width=31mm,font=\small},
  toolbox/.style={draw=figgreen,rounded corners=1mm,fill=figlightgreen,align=center,minimum height=9mm,minimum width=31mm,font=\small},
  ideabox/.style={draw=figorange,rounded corners=1mm,fill=figlightorange,align=center,minimum height=9mm,minimum width=35mm,font=\small},
  arr/.style={-{Latex[length=2.1mm]},line width=.9pt,draw=figslate}
}

\hypersetup{
  pdftitle={A 59/33 Cut-LP Guarantee for Matching Augmentation},
  pdfsubject={Approximation with respect to the standard cut relaxation for matching augmentation, with a minimum-value forest augmentation corollary},
  pdfauthor={Morteza Alimi and Tobias M\"omke},
  pdfkeywords={approximation algorithms; survivable network design; matching augmentation; forest augmentation; path augmentation; 2-edge connectivity; LP rounding; depth-first search; cut dimension},
  pdflang={en-US},
  pdfdisplaydoctitle=true
}

\title{
A $59/33$ Cut-LP Guarantee for Matching Augmentation
}
\author[1]{Morteza Alimi}
\author[1]{Tobias M\"omke}
\affil[1]{Institute of Computer Science, University of Augsburg, Germany\\
\texttt{morteza.alimi@uni-a.de}\quad\texttt{moemke@informatik.uni-augsburg.de}}
\date{}

\begin{document}
\maketitle

\begin{abstract}
The Matching Augmentation Problem (MAP) asks for a minimum-cardinality set of unit-cost edges that, together with a zero-cost matching, forms a 2-edge-connected spanning multigraph.  We study the standard cut relaxation.  Bamas, Drygala, and Svensson proposed a particularly simple LP-guided algorithm: compute an extreme optimum, run a depth-first search that prioritizes large LP coordinates, and augment the resulting DFS tree optimally.

We give a new structural analysis of
the Bamas–Drygala–Svensson LP-guided DFS algorithm. The analysis combines an exact primal–dual identity
for the residual uplink problem with a rank bound that measures fractional support relative to the unit-valued skeleton.
The result is that for every root and every deterministic tie-breaking order consistent with the LP priorities, the algorithm returns a solution of cost at most
\[
 \frac{59}{33}\,c(x^*)-\frac{25}{33}
 =\left(2-\frac7{33}\right)c(x^*)-\frac{25}{33}
 \approx1.788\,c(x^*)-0.758,
\]
where $x^*$ is an optimum of the cut LP.  Consequently, the integrality gap of the relaxation is at most $59/33\approx1.788$.  

No new algorithmic step is required; the improvement is analytical.  The exact packing certificate for the residual uplink problem yields a cost identity with a packing-slack term, while a rank theorem bounds fractional support relative to the unit-valued skeleton.  A regional classification accounts for the non-tree edges, and a two-cut identity handles self-holes.  As a direct corollary, the same $59/33\approx1.788$ bound holds for Forest Augmentation in the minimum-value regime.  The proof is self-contained apart from one theorem on the dimension of minimum-cut vectors.
\end{abstract}


\section{Introduction}\label{sec:intro}

The \emph{Matching Augmentation Problem} (MAP) is the following special case of minimum-cost 2-edge-connected spanning subgraph.  The input is a finite loopless multigraph $G=(V,E)$ with edge costs in $\{0,1\}$, and the zero-cost edges form a matching $M$.  The objective is to select as few unit-cost edge copies as possible so that, together with the free matching edges, they form a 2-edge-connected spanning submultigraph.

The standard cut relaxation is the usual LP lower bound for MAP.  \citet{BDS} proposed an especially concise LP-based algorithm: compute an optimal extreme point, use its coordinates to prioritize a depth-first search, and augment the resulting DFS tree optimally by residual uplinks.  With their fixed choice $\gamma=10^{-3}$, the bound in \citet[Section~2.2, Eq.~(2)]{BDS} is roughly $2-9.99\cdot10^{-10}$ against the cut LP; \citet[Chapter~4]{AzevedoThesis} later optimized the same analysis to roughly $2-3.92\cdot10^{-8}$.  We show that the same three-step algorithm admits a substantially sharper analysis.

\paragraph{Main result.}
For every DFS root and every deterministic tie-breaking order consistent with the LP priorities, the algorithm returns a 2-edge-connected spanning submultigraph satisfying
\begin{equation}\label{eq:intro-main}
 \ALG\le \frac{59}{33}\,c(x^*)-\frac{25}{33}
 \approx1.788\,c(x^*)-0.758.
\end{equation}
Thus the integrality gap of the standard cut LP is at most $59/33\approx1.788$.  The algorithm is unchanged; the gain comes from retaining and combining information that the previous analysis discarded.

\begin{theorem}[main theorem]\label{thm:main}
For every feasible finite loopless MAP multigraph with at least two vertices, for every DFS root, and for every fixed deterministic tie-breaking order that respects the priorities in \cref{alg:dfs}, the LP-guided DFS algorithm returns, in polynomial time, a 2-edge-connected spanning submultigraph satisfying
\[
 \ALG\le \frac{59}{33}\,c(x^*)-\frac{25}{33}
 \approx1.788\,c(x^*)-0.758,
\]
where $x^*$ is an optimum of the standard cut relaxation.  Consequently, the integrality gap of the relaxation is at most $59/33\approx1.788$.
\end{theorem}

The algorithm is a support-restricted implementation of the algorithm of \citet{BDS}: DFS is run on the positive support of $x^*$, and the augmentation step uses only positive non-tree support edges.  This restriction preserves feasibility and can only make the computed augmentation more expensive than allowing all original edges, so the theorem also applies to the unrestricted implementation.

The theorem also yields a direct corollary for Forest Augmentation.  If the zero-cost forest of an instance with at least two vertices has $k$ components, then the cut LP has value at least $k$.  Equality forces every free component to be a path and excludes positive-support heavy edges at internal path vertices.  Suppressing the free paths therefore produces a MAP instance with the same LP value, giving the same $59/33\approx1.788$ guarantee; see \cref{cor:fap-face}.

\subsection{Proof architecture}

The proof of \eqref{eq:intro-main} uses four ingredients, summarized in \cref{fig:architecture}.

\paragraph{1. Exact residual duality and packing slack.}
After the priority DFS, every non-tree support edge is an ancestor--descendant uplink.  Minimum uplink cover has an exact packing dual.  For a maximum residual packing $\Packing$, let $\Missing$ count matching edges missing from the packing, let $\OutsideMass$ be heavy DFS-tree LP mass outside the packing, and let $\PackSlack$ be the non-tree LP mass remaining after the packed fundamental-cut requirements are met.  \Cref{lem:exact-identities} gives the exact identity
\begin{equation}\label{eq:intro-exact}
 \ALG
 =2m+\frac32u+\frac12\LPex-1
 -\frac12\bigl(\Missing+\OutsideMass+\PackSlack\bigr),
\end{equation}
where $m=|M|$, $u=|V|-2m$ is the number of unmatched vertices, and $\LPex=c(x^*)-(m+u)$ is the LP excess.

\paragraph{2. Fractional support relative to the integral skeleton.}
Let $\Skeleton$ be the graph of unit-valued coordinates of $x^*$ and let $\FracHeavy$ be the fractional heavy support.
For a value $t$ let $(t)_+ := \max\{0,t\}$.
If $\Skeleton$ has $\ncomp$ connected components and $\nbridges$ bridges, then
\begin{equation}\label{eq:intro-rank}
 |\FracHeavy|\le \nbridges+(2\ncomp-3)_+.
\end{equation}
The $2\ncomp-3$ term in \eqref{eq:intro-rank} follows from \cite{LLSZ} (see \cref{thm:mincut-dimension}); tight cuts crossing a fixed integral bridge contribute only one additional quotient dimension.  The resulting support-volume inequality bounds the aggregate residual of many small non-tree copies.

\paragraph{3. A regional trichotomy.}
Every positive copy in the bundle of a packed matching edge falls into exactly one of three cases: a \emph{self-hole}, a \emph{foreign hole}, or a \emph{gated copy}.  Foreign holes are paid by missing matching edges and singleton excess.  Each gated copy is dominated by the heavy DFS branch through which it enters, and local gate overload is paid by gate mass plus fractional support volume.  The resulting regional ledger is
\begin{equation}\label{eq:intro-regional}
 14m\le33(\Missing+\OutsideMass)+16\SelfMass+17u+65\LPex,
\end{equation}
where $\SelfMass$ is the total self-hole mass.

\paragraph{4. Self-holes collapse under a two-cut identity.}
For a packed matching edge $e=p_e c_e$, let $T_{c_e}$ be the DFS subtree below $e$, let $\localSelfMass{e}$ be its self-hole mass, let $\vex{p_e}$ be the singleton excess at its parent, and let $\localPackSlack{e}$ be the slack in its packed fundamental cut.  The two-cut identity gives
\begin{equation}\label{eq:intro-selfhole}
 x^*\!\left(\delta(\{p_e\}\cup T_{c_e})\right)
 =2+\vex{p_e}+\localPackSlack{e}-2\localSelfMass{e}.
\end{equation}
For every proper such union, the cut value in \eqref{eq:intro-selfhole} is at least two, and there is at most one root exception.  Summing yields
\[
 \SelfMass\le1+\LPex+\frac{\PackSlack}{2}.
\]
Substituting this inequality into \eqref{eq:intro-regional} and then into \eqref{eq:intro-exact} proves \cref{thm:main}.

\begin{figure}[t]
\centering
\begin{tikzpicture}[node distance=8mm and 9mm]
\node[flowbox] (lp) {normalized extreme\\optimum $x^*$};
\node[flowbox,right=of lp] (dfs) {LP-guided\\depth-first search};
\node[flowbox,right=of dfs] (up) {minimum residual\\uplink cover};
\node[toolbox,below=of lp] (slack) {exact packing dual\\and slack identity};
\node[toolbox,below=of dfs] (rank) {integral skeleton\\and support rank};
\node[toolbox,below=of up] (region) {regional trichotomy\\and branch gates};
\node[ideabox,below=9mm of rank] (cut) {two-cut identity for self-holes\\$\Longrightarrow\ 59/33\approx1.788$};
\draw[arr] (lp)--(dfs); \draw[arr] (dfs)--(up);
\draw[arr] (lp)--(slack); \draw[arr] (dfs)--(rank); \draw[arr] (up)--(region);
\draw[arr] (slack)--(rank); \draw[arr] (rank)--(region);
\draw[arr] (slack)--(cut); \draw[arr] (region)--(cut);
\end{tikzpicture}
\caption{The algorithm has three steps; the analysis uses three certificates and one cut identity.  Colors distinguish algorithmic objects from proof objects.
}
\label{fig:architecture}
\end{figure}
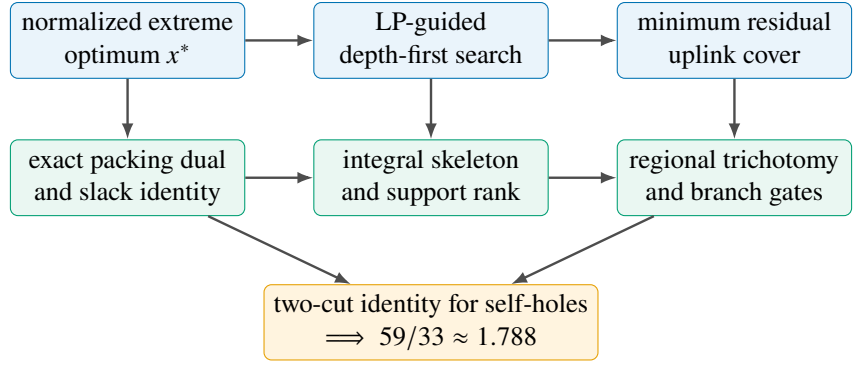

\paragraph{Simplicity and scope.}
The algorithm consists of one LP computation, one priority DFS, and one greedy exact augmentation.  The proof imports only the minimum-cut dimension theorem, stated as \cref{thm:mincut-dimension}; the remaining arguments are exact packing duality, elementary cut identities, an extreme-point rank argument, and local charging.  The factor-revealing local inequality is elementary.  The implementation remains the short LP--DFS--uplink pipeline, while the analysis proves a statement about the standard relaxation.

\paragraph{Organization.}
\Cref{sec:related} surveys and positions the result.  \Cref{sec:algorithm} defines the algorithm and proves exact uplink cover--packing duality.  \Cref{sec:slack} derives the exact packing-slack identity.  \Cref{sec:rank} proves the integral-skeleton rank theorem and its support-volume consequence.  \Cref{sec:regional} develops the three regional channels and the deficit bound.  \Cref{sec:selfhole} proves the self-hole theorem, and \cref{sec:main-proof} completes the approximation analysis.  \Cref{sec:fap-face} derives the minimum-value Forest Augmentation corollary.  \Cref{sec:discussion} discusses the scope of the analysis.  \Cref{app:technical} collects multigraph and boundary-case details.

\section{Related work and positioning}\label{sec:related}

\paragraph{MAP against the integral optimum.}
The first ratio below two for MAP was $7/4=1.750$ by \citet{Cheriyan7over4}, based on a reduction to well-structured instances and ear augmentation.  \citet{Cheriyan5over3} improved this to $5/3\approx1.667$ using new local augmentation techniques.  The best established guarantee is $13/8=1.625$ by \citet{Garg13over8}, whose central ingredients are an approximation-preserving structural reduction and repeated simultaneous contractions. While this is numerically smaller than $59/33\approx1.788$, it does not imply a guarantee against the value of the standard cut relaxation.

\paragraph{MAP against the cut LP.}
\citet{BDS} introduced the LP-guided DFS algorithm studied here and proved both a better-than-two upper bound and a $4/3\approx1.333$ integrality-gap lower bound for the standard cut relaxation.  With their fixed choice $\gamma=10^{-3}$, the bound in Section~2.2, Eq.~(2), is about $2-9.99\cdot10^{-10}\approx1.999999999$.  \citet[Chapter~4]{AzevedoThesis} optimized the same parameterized analysis to about $2-3.92\cdot10^{-8}\approx1.999999961$.  The present $59/33\approx1.788$ guarantee is obtained by a different analysis of the same algorithm, rather than by retuning those parameters.  \Cref{tab:prior-work} summarizes these guarantees.

\begin{table}[t]
\centering
\small
\caption{Selected MAP guarantees.  The decimal column makes the numerical ordering explicit; the last column records the benchmark, because ratios against the integral optimum and against the cut LP answer different questions.}
\label{tab:prior-work}
\begin{tabularx}{\textwidth}{@{}>{\raggedright\arraybackslash}p{.25\textwidth}>{\centering\arraybackslash}p{.15\textwidth}>{\centering\arraybackslash}p{.16\textwidth}>{\raggedright\arraybackslash}X@{}}
\toprule
Result & Guarantee & Rough decimal & Benchmark and principal ingredient\\
\midrule
\citet{Cheriyan7over4} & $7/4$ & $1.750$ & integral optimum; structural reduction and ear augmentation\\
\citet{Cheriyan5over3} & $5/3$ & $1.667$ & integral optimum; local augmentation\\
\citet{Garg13over8} & $13/8$ & $1.625$ & integral optimum; simultaneous contractions\\
\citet{BDS}, fixed parameters & $2-9.99\cdot10^{-10}$ & $1.999999999$ & cut LP; LP-guided DFS and exact tree augmentation\\
\citet{AzevedoThesis} & $2-3.92\cdot10^{-8}$ & $1.999999961$ & cut LP; optimized analysis of the same algorithm\\
This paper & $59/33$ & $1.788$ & cut LP; exact residual accounting and support rank\\
\bottomrule
\end{tabularx}
\end{table}

\paragraph{Forest and path augmentation.}
For the broader Forest Augmentation Problem, \citet{GrandoniFAP} first broke the factor-two barrier, obtaining $1.9973$ for FAP through a reduction to Path Augmentation and a $1.9913$ PAP algorithm.  \citet{HommelsheimFAP} subsequently obtained $1.9412$ for PAP and $1.9955$ for FAP using a new relaxation and a stronger reduction to structured instances.  These guarantees are measured against the integral optimum and rely on global structural transformations.  Our FAP corollary concerns the narrower minimum-LP-value face, where equality in the elementary lower bound forces the free forest itself to be a path forest and permits an exact suppression to MAP.

\paragraph{Broader augmentation context.}
MAP is a special $0/1$-cost 2-edge-connected spanning-subgraph problem, whereas Tree and Connectivity Augmentation start from a connected free subgraph and add links.  \citet{CecchettoCAP} developed unified methods that transfer ideas between those connected-base models.  In our setting, the priority DFS creates a rooted uplink residual instance, so the augmentation itself is exactly solvable; the novelty lies in relating its exact packing dual back to the original cut LP.

\paragraph{Technical ingredients.}
The residual rooted uplink problem is a special case of the graph-augmentation setting studied by Frederickson and JáJá~\citep{FredericksonJaja}; the precise cover--packing relation used here is proved directly in \cref{prop:minmax}.  Our polyhedral argument uses the minimum-cut dimension theorem of \citet{LLSZ}, stated as \cref{thm:mincut-dimension}.  The new work lies in connecting these two ingredients to the unit-valued skeleton of an extreme MAP cut LP solution and to the regional accounting of the LP-guided DFS.

\section{The cut LP and the LP-guided DFS algorithm}\label{sec:algorithm}

Let $G=(V,E)$ be a finite loopless multigraph with $|V|\ge2$.  The zero-cost edges form a matching $M$: they are pairwise vertex-disjoint and there is at most one zero-cost edge on each matched pair.  Every edge in $E\setminus M$ has unit cost.  We call the matching edges \emph{light} and the remaining edges \emph{heavy}.  Parallel heavy copies are encoded explicitly as separate edges and separate LP coordinates.  For a nonempty proper set $S\subsetneq V$, let $\delta(S)$ be the multiset of edges with exactly one endpoint in $S$, and write $x(F)=\sum_{e\in F}x_e$.

The standard cut relaxation is
\begin{equation}\label{lp:map}
\begin{aligned}
 \min\quad & c(x):=\sum_{e\in E\setminus M}x_e\\
 \text{s.t.}\quad & x(\delta(S))\ge2 &&(\emptyset\ne S\subsetneq V),\\
 &0\le x_e\le1 &&(e\in E).
\end{aligned}
\end{equation}
Every 2-edge-connected spanning submultigraph gives a feasible integral solution.  The relaxation can be solved in polynomial time by minimum-cut separation.

Because matching coordinates have zero objective coefficient, increasing any of them preserves feasibility and objective value.  The rank argument later needs the normalized vector itself to be extreme.

\begin{lemma}[extreme normalization]\label{lem:normalization}
There is an optimal extreme point $x^*$ of \eqref{lp:map} satisfying
\[
 x^*_e=1\qquad(e\in M).
\]
\end{lemma}
\begin{proof}
The intersection of the optimal face with the face determined by $x_e=1$ for all $e\in M$ is nonempty: starting from any optimum, increasing a matching coordinate preserves every cut constraint and does not change the objective.  Choose a vertex of this intersection.  An intersection of faces is a face of the original polytope, and every vertex of a face is a vertex of the original polytope.

For algorithmic purposes, fix $x_M=\one$, solve the resulting heavy-coordinate LP by minimum-cut separation, and choose a basic optimum by polynomially many lexicographic optimizations over the optimal face, using the same separation oracle.
\end{proof}

Fix such a point.  Set
\[
 m:=|M|,\qquad u:=|V|-2m,\qquad n:=|V|=2m+u,
\]
and write
\[
 c^*:=c(x^*)=m+u+\LPex.
\]
The LP excess $\LPex$ is nonnegative: summing the singleton constraints gives $x^*(E)\ge n$, while $x^*(M)=m$.

For a vertex $v$, define the required heavy degree
\[
 \reqdeg{v}:=\begin{cases}
 1,&v\text{ is matched},\\
 2,&v\text{ is unmatched},
 \end{cases}
\]
and the singleton excess
\[
 \vex{v}:=x^*(\delta(v)\setminus M)-\reqdeg{v}\ge0.
\]
The total heavy degree is
\[
 \sum_{v\in V}x^*(\delta(v)\setminus M)=2c^*.
\]
Consequently,
\begin{equation}\label{eq:eps-budget}
 \sum_{v\in V}\vex{v}=2\LPex.
\end{equation}

The notation used throughout the analysis is summarized in \cref{tab:notation}.

\begin{table}[t]
\centering
\small
\caption{Mnemonic notation used throughout the analysis.}
\label{tab:notation}
\begin{tabularx}{\textwidth}{@{}>{\raggedright\arraybackslash}p{.24\textwidth}X@{}}
\toprule
Symbol & Meaning\\
\midrule
$m,u$, and $n=2m+u$ & number of matching edges, unmatched vertices, and total vertices\\
$\LPex,\vex{v}$ & global LP excess and singleton excess at vertex $v$\\
$\Uplinks,\Uplinks(f)$ & positive non-tree uplinks and the bundle covering tree edge $f$\\
$\Packing,\PackM,\PackH$ & maximum packing and its matching/heavy parts\\
$\Missing,\OutsideMass,\PackSlack$ & missing matching edges, outside tree mass, and packing slack\\
$\Saving=\Missing+\OutsideMass+\PackSlack$ & exact saving subtracted from the baseline cost\\
$\Skeleton,\FracHeavy,\ncomp,\nbridges$ & integral skeleton, fractional heavy support, and the skeleton's component and bridge counts\\
$\SelfMass,\GateOver$ & total self-hole mass and total gate overload\\
\bottomrule
\end{tabularx}
\end{table}

\begin{definition}[LP-guided DFS algorithm]\label{alg:dfs}
\leavevmode
\begin{enumerate}[label=\arabic*.,leftmargin=7mm]
\item Compute an optimal extreme point $x^*$ as in \cref{lem:normalization}.
\item Fix an arbitrary root and an arbitrary deterministic tie-breaking order.  In the support graph
\[
 G^+:=\bigl(V,\{e:x^*_e>0\}\bigr),
\]
run DFS.  If the current vertex has an edge of $M$ to an unvisited vertex, DFS takes that edge; otherwise it takes an eligible heavy edge of maximum $x^*$-value, using that order to break ties.  Let $T$ be the DFS tree.
\item Compute a minimum-cardinality augmentation of $T$ whose links are positive non-tree support copies, and return its union with $T$.
\end{enumerate}
\end{definition}

The support graph is connected because every proper cut has positive $x^*$-mass.  All matching edges belong to $T$: when one endpoint is first discovered, its mate is still unvisited, so the priority rule selects the matching edge.  Thus $T$ contains exactly
\[
 n-1-m=m+u-1
\]
heavy edges.

Every non-tree edge of an undirected DFS joins an ancestor and a descendant.  Root $T$ at the DFS root, orient each non-tree support edge from its descendant \emph{source} to its ancestor \emph{top endpoint}, and call it an \emph{uplink}.  Let
\[
 \Uplinks:=E(G^+)\setminus E(T)
\]
be the heavy non-tree support edges.  An uplink covers the tree path between its endpoints.  For a tree edge $f$, let
\[
 \Uplinks(f):=\{\ell\in \Uplinks:f\text{ lies on the tree path covered by }\ell\}.
\]
The fundamental cut of $f$ contains no other tree edge, so
\begin{equation}\label{eq:fundamental}
 x^*(\Uplinks(f))\ge2-x^*_f\ge1,
\end{equation}
where the last inequality uses $x^*_f\le1$.  Hence the support-restricted residual instance is feasible.

\subsection{Exact uplink cover--packing duality}

A set $\Packing\subseteq E(T)$ is an \emph{uplink packing} if no uplink covers two edges of $\Packing$.  
The residual rooted uplink augmentation problem is a special case of the graph-augmentation setting studied by Frederickson and JáJá~\citep{FredericksonJaja}.
\Cref{fig:packing} illustrates an uplink packing and the corresponding disjoint packing bundles.

\begin{proposition}[uplink cover--packing theorem]\label{prop:minmax}
For a rooted tree and a set of uplinks that together cover all tree edges, the minimum number of uplinks covering all tree edges equals the maximum cardinality of an uplink packing.  Both objects can be found in polynomial time.
\end{proposition}
\begin{proof}
Repeatedly choose a deepest uncovered tree edge $f$.  Among the uplinks covering $f$, choose one whose top endpoint is highest.  Add the chosen uplink to the cover, add $f$ to $\Packing$, and mark every tree edge covered by the chosen uplink.

The selected uplinks cover the tree when the procedure stops.  We claim that $\Packing$ is a packing.  Suppose an uplink $\ell'$ covered two selected edges $f_i,f_j$, with $f_i$ selected first.  Both edges lie on the ancestor--descendant path covered by $\ell'$ and are therefore comparable.  The edge $f_i$ is the deeper one: otherwise $f_j$ was still uncovered when $f_i$ was chosen and would have been selected instead.  At that iteration, $\ell'$ was a candidate covering $f_i$.  The greedy uplink chosen for $f_i$ has top endpoint at least as high as that of $\ell'$, and therefore also covers $f_j$.  Thus $f_j$ would have been marked and could not later enter $\Packing$, a contradiction.

Every uplink cover has size at least $|\Packing|$, because one uplink covers at most one packed edge.  The greedy cover has one uplink per selected edge, so cover and packing sizes are equal and optimal.
\end{proof}

The algorithm returns the $m+u-1$ heavy tree edges together with a minimum uplink cover.  Every tree edge is covered by a selected uplink and therefore lies on a fundamental cycle; every selected uplink lies on the same cycle.  Hence the returned connected multigraph has no bridge and is 2-edge-connected.  If $\Packing$ is a maximum packing, then
\begin{equation}\label{eq:alg-basic}
 \ALG=(m+u-1)+|\Packing|.
\end{equation}
Under the explicit edge-copy encoding, minimum-cut separation and polynomially many lexicographic optimizations over the optimal face compute a rational basic optimum using exact arithmetic in time polynomial in the binary input length.  The deterministic tie-breaking order fixes the DFS execution.  The greedy procedure in \cref{prop:minmax} makes at most $|E(T)|$ selections and can scan $\Uplinks$ at each iteration, so the augmentation and hence the whole algorithm are polynomial-time.

\begin{figure}[t]
\centering
\begin{tikzpicture}[scale=.92]
  \node[vnode] (rootnode) at (0,3.5) {};
  \node[vnode] (a) at (-1.5,2.2) {};
  \node[vnode] (b) at (1.5,2.2) {};
  \node[vnode] (c) at (-2.0,.7) {};
  \node[vnode] (d) at (-.8,.7) {};
  \node[vnode] (e) at (.8,.7) {};
  \node[vnode] (f) at (2.0,.7) {};
  \draw[heavyedge] (rootnode)--(a) (rootnode)--(b) (a)--(c) (b)--(f);
  \draw[lightedge] (a)--(d) (b)--(e);
  \draw[packingedge] (a)--(d);
  \draw[packingedge] (rootnode)--(b);
  \draw[backedge,bend left=25] (d) to (rootnode);
  \draw[backedge,bend right=23] (e) to (rootnode);
  \draw[backedge,bend right=37] (f) to (rootnode);
  \draw[heavyedge] (2.8,3.15)--(3.55,3.15);
  \node[font=\small,anchor=west] at (3.75,3.15) {heavy tree edge};
  \draw[lightedge] (2.8,2.60)--(3.55,2.60);
  \node[font=\small,anchor=west] at (3.75,2.60) {matching tree edge};
  \draw[packingedge] (2.8,2.05)--(3.55,2.05);
  \node[font=\small,anchor=west] at (3.75,2.05) {edge in $\Packing$};
  \draw[backedge] (2.8,1.50)--(3.55,1.50);
  \node[font=\small,anchor=west] at (3.75,1.50) {uplink};
\end{tikzpicture}
\caption{An uplink packing $\Packing$, the exact certificate for the residual cover problem.  Gray and double-green tree edges are heavy and matching edges, respectively, while blue overlays mark the two packed edges.  Each dashed red arrow joins a vertex to an ancestor and is therefore an uplink; the two packing bundles are disjoint because no uplink covers both blue edges.}
\label{fig:packing}
\end{figure}
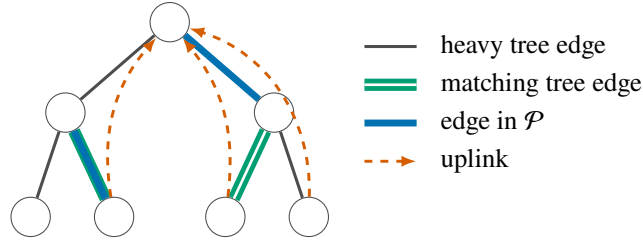

\section{The exact packing-slack identity}\label{sec:slack}

Fix a maximum packing $\Packing$ and split it into
\[
 \PackM:=\Packing\cap M,
 \qquad
 \PackH:=\Packing\cap(T\setminus M).
\]
Write
\[
 \npacked:=|\PackM|,
 \qquad
 \Missing:=m-\npacked,
\]
and let
\[
 \HeavyTree:=T\setminus M,
 \qquad
 \TreeMass:=x^*(\HeavyTree),
 \qquad
 \PackedMass:=x^*(\PackH),
 \qquad
 \OutsideMass:=\TreeMass-\PackedMass.
\]
Thus $\Missing$ counts matching edges absent from the residual packing, and $\OutsideMass$ is heavy tree LP mass outside the packing.

Because $\Packing$ is a packing, the bundle sets $\Uplinks(f)$, $f\in \Packing$, are pairwise disjoint.  Define the \emph{packing slack}
\begin{equation}\label{eq:sigma-def}
 \PackSlack:=c^*-\TreeMass-\sum_{f\in \Packing}(2-x^*_f).
\end{equation}
By \eqref{eq:fundamental}, $\PackSlack\ge0$.

The following decomposition gives $\PackSlack$ a direct interpretation.  Let
\[
 \Uplinks_{\mathrm{out}}:=\Uplinks\setminus\bigcup_{f\in \Packing}\Uplinks(f),
 \qquad
 \localPackSlack{f}:=x^*(\Uplinks(f))-(2-x^*_f)\ge0.
\]
Then
\begin{equation}\label{eq:sigma-decomp}
 \PackSlack=x^*(\Uplinks_{\mathrm{out}})+\sum_{f\in \Packing}\localPackSlack{f}.
\end{equation}
It consists of non-tree mass outside all packing bundles plus surplus beyond the fundamental-cut requirement inside those bundles.  For a packed matching edge $e\in \PackM$, write
\[
 \localPackSlack{e}=x^*(\Uplinks(e))-1.
\]

\begin{lemma}[exact packing identities]\label{lem:exact-identities}
The following identities hold:
\begin{align}
 2|\Packing|&=c^*+\npacked-\OutsideMass-\PackSlack,\label{eq:packing-exact}\\
 \Missing+u+\LPex&=\OutsideMass+2|\PackH|+\PackSlack.\label{eq:second-exact}
\end{align}
Consequently, for
\[
 \Saving:=\Missing+\OutsideMass+\PackSlack,
\]
the algorithmic cost is exactly
\begin{equation}\label{eq:alg-exact}
 \ALG=2m+\frac32u+\frac12\LPex-1-\frac12\Saving.
\end{equation}
\end{lemma}
\begin{proof}
Since $x^*(\Packing)=\npacked+\PackedMass$, the definition \eqref{eq:sigma-def} of $\PackSlack$ gives
\[
 \PackSlack=c^*-\TreeMass-2|\Packing|+\npacked+\PackedMass=c^*+\npacked-\OutsideMass-2|\Packing|,
\]
which is \eqref{eq:packing-exact}.  Writing $|\Packing|=\npacked+|\PackH|$ and $c^*=m+u+\LPex$ in the same identity gives \eqref{eq:second-exact}.  Substituting \eqref{eq:packing-exact} into \eqref{eq:alg-basic} proves \eqref{eq:alg-exact}.
\end{proof}

The proof will lower-bound the exact saving $\Saving$.  The three terms have different interpretations:
\begin{center}
\begin{tabular}{@{}cl@{}}
$\Missing$ & matching edges missing from the packing;\\
$\OutsideMass$ & heavy DFS-tree mass outside the packing;\\
$\PackSlack$ & non-tree LP mass unused by the packed fundamental cuts.
\end{tabular}
\end{center}

A useful consequence of \eqref{eq:second-exact} is
\begin{equation}\label{eq:tau-upper}
 \TreeMass\le\frac12(\Missing+u+\LPex+\OutsideMass).
\end{equation}
Indeed, $\PackedMass\le|\PackH|$ and $\TreeMass=\OutsideMass+\PackedMass$.

\subsection{A six-vertex example}\label{subsec:envelope}

We review the known $4/3$ gap instance for the standard cut relaxation~\citep[Theorem~2 and Figure~1(a)]{BDS} as it illustrates the exact identity.  Let the three light edges be $ab$, $cd$, and $ef$, and let the heavy edges form the triangles induced by $\{a,c,e\}$ and $\{b,d,f\}$, see \Cref{fig:envelope}.  Give every heavy edge value $1/2$.  This is a feasible LP solution of value $c^*=3$.

Root DFS at $a$ and use a deterministic tie-breaking order producing the tree path
\[
 a-b-d-c-e-f.
\]
The heavy tree edges are $bd$ and $ce$, so $\TreeMass=1$.  The four remaining heavy edges are uplinks.  A minimum cover has size two, and a maximum packing is $\Packing=\{ab,ef\}$.  Hence
\[
 \Missing=1,
 \qquad \OutsideMass=1,
 \qquad \PackSlack=0,
 \qquad \ALG=2+2=4.
\]
The exact identity gives
\[
 4=2m-1-\frac12(\Missing+\OutsideMass+\PackSlack)=5-1.
\]
Moreover, the two uplinks from $a$ into the left triangle are self-holes of the packed edge $ab$, of total mass $\SelfMass=1$.  Thus this example also attains equality in the self-hole theorem of \cref{lem:selfhole} through its unique root exception; see \cref{rem:root-exception}.

\begin{figure}[H]
\centering
\begin{tikzpicture}[scale=.94,every node/.style={font=\small}]
\node[vnode] (a) at (0,2) {$a$};
\node[vnode] (e) at (0,0) {$e$};
\node[vnode] (c) at (1.45,1) {$c$};
\node[vnode] (b) at (5.2,2) {$b$};
\node[vnode] (f) at (5.2,0) {$f$};
\node[vnode] (d) at (3.75,1) {$d$};
\draw[lightedge] (a)--(b) node[midway,above=3pt,fill=white,inner sep=1pt] {$ab$};
\draw[lightedge] (c)--(d) node[midway,above=3pt,fill=white,inner sep=1pt] {$cd$};
\draw[lightedge] (e)--(f) node[midway,below=3pt,fill=white,inner sep=1pt] {$ef$};
\draw[backedge,bend left=8] (e) to (a);
\draw[backedge,bend left=9] (c) to (a);
\draw[gateedge] (c)--(e);
\draw[backedge,bend right=8] (f) to (b);
\draw[gateedge] (b)--(d);
\draw[backedge,bend left=9] (f) to (d);
\draw[lightedge] (6.05,1.85)--(6.80,1.85);
\node[anchor=west] at (7.00,1.85) {matching edge};
\draw[gateedge] (6.05,1.25)--(6.80,1.25);
\node[anchor=west] at (7.00,1.25) {heavy DFS-tree edge};
\draw[backedge] (6.05,.65)--(6.80,.65);
\node[anchor=west] at (7.00,.65) {residual uplink};
\node[align=left,anchor=west] at (6.05,.05) {Every heavy edge has\\LP value $1/2$.};
\end{tikzpicture}
\caption{The six-vertex envelope instance.  The double-green edges form the matching $M=\{ab,cd,ef\}$, and every heavy edge has LP value $1/2$.  Rooted DFS at $a$ follows the path $a-b-d-c-e-f$: the thick dotted edges $bd$ and $ce$ (orange) are its heavy tree edges, while the four dashed arrows (red) are the residual uplinks.}
\label{fig:envelope}
\end{figure}
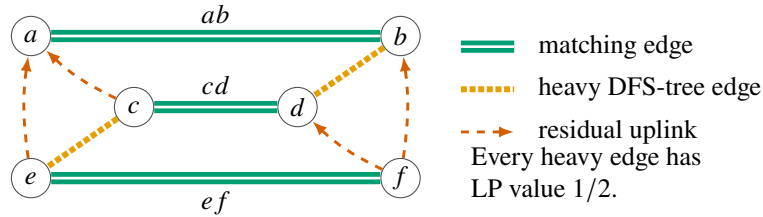

\section{Fractional support relative to the integral skeleton}\label{sec:rank}

The regional argument must bound the aggregate deficit of many small gated copies.  The global resource is a support-volume bound obtained from the unit-valued coordinates of the extreme point.

Let
\[
 \Skeleton:=\{e\in E:x^*_e=1\}
\]
be the \emph{integral skeleton}; it contains $M$.  Let
\[
 \FracHeavy:=\{e\in E\setminus M:0<x^*_e<1\}
\]
be the fractional heavy support.  Write $\ncomp$ for the number of connected components of $(V,\Skeleton)$ and $\nbridges$ for its number of bridge edge copies.

In the following theorem, the edge set of a nonnegatively weighted graph is understood to be its positive-weight support.
The following is \citet[Theorem~1]{LLSZ}, stated in the notation needed here.

\begin{theorem}[minimum-cut dimension]\label{thm:mincut-dimension}
Let $Q=(W,F)$ be an undirected graph with nonnegative edge weights, let $q:=|W|\ge2$, and let $\mathcal C_{\min}(Q)$ be the family of its global minimum cuts.  For $C\in\mathcal C_{\min}(Q)$, let $\chi_F(C)\in\{0,1\}^F$ denote its characteristic vector.  Then
\[
 \dim\operatorname{span}\{\chi_F(C):C\in\mathcal C_{\min}(Q)\}\le 2q-3.
\]
\end{theorem}

For comparison, \citet{BoydPulleyblank} proved the classical TSP-specific bound that every extreme point of the subtour-elimination polytope on $q$ vertices has support of size at most $2q-3$.  For MAP, \citet[Lemma~1]{BDS} proved that an extreme cut-LP solution has at most $2n-1$ fractional edges.  The estimate below instead records how the fractional support depends on the component and bridge structure of the unit-valued skeleton.

The next theorem shows how this cut dimension interacts with the integral skeleton.

\begin{theorem}[integral-skeleton rank]\label{thm:rank}
\[
 |\FracHeavy|\le \nbridges+(2\ncomp-3)_+.
\]
\end{theorem}

\begin{proof}
For a vertex set $S$, write
\[
 r(S):=\chi_{\FracHeavy}(\delta(S))\in\{0,1\}^{\FracHeavy}.
\]
Because $x^*$ is extreme and every coordinate in $\FracHeavy$ lies strictly between its bounds, the rows $r(S)$ of the tight cuts span $\mathbb R^{\FracHeavy}$.  Otherwise there is a nonzero direction $d$ supported on $\FracHeavy$ and orthogonal to every tight cut row; then $x^*\pm\eta d$ remains feasible for sufficiently small $\eta>0$, contradicting extremality.

A relevant tight cut, meaning one with $r(S)\ne0$, crosses at most one edge of $\Skeleton$: each such edge contributes one to a cut of total value two, and a tight cut crossing two skeleton edges has zero $\FracHeavy$-row.  If it crosses one skeleton edge $h$, then $h$ is a bridge of $\Skeleton$, because the intersection of any cut with a cycle has even cardinality, counting copies.

Let $\mathcal T_0$ be the relevant tight sets crossing no skeleton edge.  For each bridge $h$ of $\Skeleton$, let $\mathcal T_h$ be the relevant tight sets whose only crossing skeleton edge is $h$, and set
\[
 W_0:=\operatorname{span}\{r(S):S\in\mathcal T_0\},
 \qquad
 W_h:=W_0+\operatorname{span}\{r(S):S\in\mathcal T_h\}.
\]
These families contain every relevant tight row.
The two types of tight cuts used in the proof are illustrated in \cref{fig:rank}.

\paragraph{Type-zero rows.}
Every set in $\mathcal T_0$ is a union of components of $\Skeleton$.  Contract the $\ncomp$ components and, for each pair of resulting supervertices, replace its parallel fractional copies by one quotient edge whose weight is their total $x^*$-value.  Every quotient edge has positive weight, since it represents a nonempty set of fractional support copies.  Fractional copies within one skeleton component become loops and never occur in a type-zero row; zero-valued input copies are not coordinates of $\FracHeavy$ and are irrelevant.  Every quotient cut has weight at least two, and any cut represented by $\mathcal T_0$ has weight two and is therefore a global minimum cut.

The quotient cut vector pulls back to the corresponding row in $\mathbb R^{\FracHeavy}$ by duplicating its coordinate on all original parallel copies and inserting zero on within-component copies.  This linear duplication map is injective, so it preserves the dimension of the quotient row space.  By \cref{thm:mincut-dimension},
\[
 \dim W_0\le (2\ncomp-3)_+;
\]
for $\ncomp=1$, the family $\mathcal T_0$ is empty.

\paragraph{Rows crossing one fixed bridge.}
Fix a bridge $h=pq$ of $\Skeleton$ and orient every set in $\mathcal T_h$ so that it contains $p$ and excludes $q$.  These sets form a chain.  Otherwise two of them, say $X$ and $Y$, cross.  With
\[
 A=X\setminus Y,
 \quad B=Y\setminus X,
 \quad C=X\cap Y,
 \quad D=V\setminus(X\cup Y),
\]
we have $p\in C$, $q\in D$, and
\[
 x^*(\delta(X))+x^*(\delta(Y))
 =x^*(\delta(A))+x^*(\delta(B))+2x^*(E(C,D)).
\]
The left side is four, whereas the right side is at least $2+2+2x^*_h=6$, a contradiction.

Let $X\subsetneq Y$ be two sets in this chain, and put $R:=Y\setminus X$ and $Q:=V\setminus Y$.  No skeleton edge crosses $R$.  Define
\[
 \mu_{XR}:=x^*_{\FracHeavy}(E(X,R)),
 \qquad
 \mu_{XQ}:=x^*_{\FracHeavy}(E(X,Q)),
 \qquad
 \mu_{RQ}:=x^*_{\FracHeavy}(E(R,Q)).
\]
Tightness of $X$ and $Y$ gives $\mu_{XR}+\mu_{XQ}=1$ and $\mu_{XQ}+\mu_{RQ}=1$, while feasibility of the type-zero cut $R$ gives $\mu_{XR}+\mu_{RQ}\ge2$.  Hence $\mu_{XQ}=0$ and $\mu_{XR}=\mu_{RQ}=1$.  Thus $R\in\mathcal T_0$, and positivity of every coordinate in $\FracHeavy$ implies that no fractional support edge joins $X$ to $Q$.  Consequently,
\begin{equation}\label{eq:bridge-relation}
 r(X)+r(Y)=r(R)\in W_0.
\end{equation}
By \eqref{eq:bridge-relation}, all rows of $\mathcal T_h$ have the same one-dimensional span modulo $W_0$, up to sign, and $\dim(W_h/W_0)\le1$.

Since the relevant tight rows span $\mathbb R^{\FracHeavy}$ and there are $\nbridges$ bridge classes, applying \cref{thm:mincut-dimension} we obtain
\[
 |\FracHeavy|
 \le \dim W_0+\sum_h\dim(W_h/W_0)
 \le (2\ncomp-3)_+ +\nbridges.
\]
\end{proof}

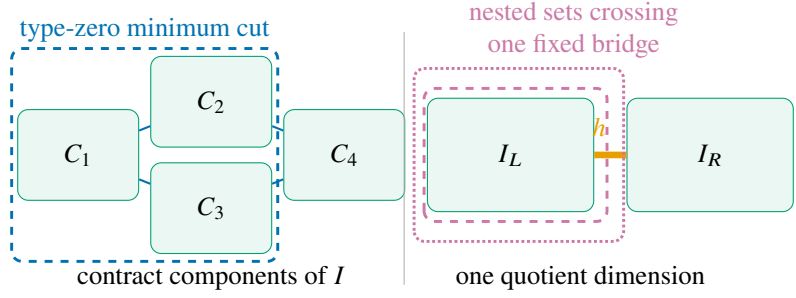
\begin{figure}[t]
\centering
\begin{tikzpicture}[scale=.88,every node/.style={font=\small}]
% Left panel: type-zero cuts
\node[draw=figgreen,fill=figlightgreen,rounded corners,minimum width=16mm,minimum height=12mm] (c1) at (0,1.5) {$C_1$};
\node[draw=figgreen,fill=figlightgreen,rounded corners,minimum width=16mm,minimum height=12mm] (c2) at (2.0,2.3) {$C_2$};
\node[draw=figgreen,fill=figlightgreen,rounded corners,minimum width=16mm,minimum height=12mm] (c3) at (2.0,.7) {$C_3$};
\node[draw=figgreen,fill=figlightgreen,rounded corners,minimum width=16mm,minimum height=12mm] (c4) at (4.0,1.5) {$C_4$};
\draw[figblue,thick] (c1)--(c2) (c1)--(c3) (c2)--(c4) (c3)--(c4);
\draw[figblue,dashed,rounded corners,line width=1.1pt] (-1.0,-.1) rectangle (3.0,3.1);
\node[figblue] at (1.0,3.38) {type-zero minimum cut};
\node at (2,-.35) {contract components of $\Skeleton$};
% separator
\draw[figslate!45] (4.9,-.45)--(4.9,3.35);
% Right panel: bridge cuts
\node[draw=figgreen,fill=figlightgreen,rounded corners,minimum width=22mm,minimum height=15mm] (l) at (6.5,1.5) {$I_L$};
\node[draw=figgreen,fill=figlightgreen,rounded corners,minimum width=22mm,minimum height=15mm] (rightskel) at (9.5,1.5) {$I_R$};
\draw[figorange,line width=2.5pt] (l)--(rightskel) node[pos=.15,above=2pt] {$h$};
\draw[figpurple,dashed,rounded corners,line width=1.1pt] (5.2,.5) rectangle (7.95,2.5);
\draw[figpurple,densely dotted,rounded corners,line width=1.1pt] (5.05,.2) rectangle (8.2,2.8);
\node[figpurple,align=center] at (7.45,3.38) {nested sets crossing\\one fixed bridge};
\node at (7.55,-.35) {one quotient dimension};
\end{tikzpicture}
\caption{The rank theorem has two parts.  Tight cuts avoiding the integral skeleton become minimum cuts after contracting its components; tight cuts crossing one fixed integral bridge form a chain and add at most one dimension modulo the first family.}
\label{fig:rank}
\end{figure}

\begin{lemma}[support-volume bound]\label{lem:support-volume}
For every set $\mathcal A\subseteq\Uplinks$,
\[
 \sum_{e\in\mathcal A}(1-x^*_e)\le m+\TreeMass.
\]
\end{lemma}
\begin{proof}
Integral copies contribute zero, so the left side is at most
\[
 \sum_{e\in \FracHeavy\setminus T}(1-x^*_e).
\]
Let $\mu_{\mathrm{frac}}:=x^*(\FracHeavy)$ and let $i_T$ be the number of integral heavy tree edges.  Since $T$ contains exactly $m+u-1$ heavy edges,
\[
 |\FracHeavy\cap T|=m+u-1-i_T,
 \qquad
 x^*(\FracHeavy\cap T)=\TreeMass-i_T.
\]
Therefore
\begin{equation}\label{eq:R-count}
 \sum_{e\in\mathcal A}(1-x^*_e)
 \le |\FracHeavy|-m-u+1-\mu_{\mathrm{frac}}+\TreeMass.
\end{equation}

If $\ncomp\ge2$, summing the cut inequalities of the $\ncomp$ components of $\Skeleton$ gives $\mu_{\mathrm{frac}}\ge\ncomp$.  Deleting the $\nbridges$ bridges of $\Skeleton$ creates exactly $\nbridges+\ncomp$ nonempty components, so $\nbridges+\ncomp\le n$.  By \cref{thm:rank}, the right side of \eqref{eq:R-count} is at most
\[
 \nbridges+2\ncomp-3-m-u+1-\ncomp+\TreeMass
 \le n-m-u-2+\TreeMass
 =m-2+\TreeMass.
\]
If $\ncomp=1$, \cref{thm:rank} gives $|\FracHeavy|\le \nbridges$ and $\nbridges+1\le n$, yielding at most $m+\TreeMass$.  We record the uniform form because it is the only one needed later; the sharper $m-2+\TreeMass$ bound is available whenever the integral skeleton is disconnected.
\end{proof}

\section{A regional classification of packed matching bundles}\label{sec:regional}

Fix $e\in \PackM$, and orient the matching tree edge as
\[
 e=p_e c_e,
\]
where $p_e$ is the parent of $c_e$.  Every $\ell\in \Uplinks(e)$ has ancestor endpoint (top endpoint) $a$ and descendant endpoint (source) $z$.  Let $ac$ be the first tree edge on the path from $a$ toward $z$.  Exactly one of the following occurs.
\begin{enumerate}
\item \emph{Self-hole:} $ac=e$.  Then $a=p_e$.  Let $\localSelfMass{e}$ be the total self-hole mass in $\Uplinks(e)$, and put
\[
 \SelfMass:=\sum_{e\in \PackM}\localSelfMass{e}.
\]
\item \emph{Foreign hole:} $ac\in M\setminus\{e\}$.  Since $\ell$ already covers the packed edge $e$, the packing property implies $ac\notin \PackM$.
\item \emph{Gated copy:} $ac$ is heavy.  We call $ac$ the branch gate of $\ell$.
\end{enumerate}
The three cases are illustrated in \cref{fig:channels}.
These cases are disjoint and exhaustive.  The first tree edge below the top endpoint is either $e$, another matching edge, or heavy.  If it is another matching edge that also belongs to $\Packing$, then the uplink covers two packing edges, contradicting the definition of $\Packing$; hence every foreign hole enters through a matching edge missing from the packing.

\begin{figure}[t]
\centering
\begin{tikzpicture}[scale=.92,every node/.style={font=\small}]
\node[vnode] (u) at (0,2.8) {$p_e$}; \node[vnode] (v) at (0,1.8) {$c_e$}; \node[vnode] (z) at (0,0) {$z$};
\draw[lightedge] (u)--(v); \draw[packingedge] (u)--(v) node[midway,right=3pt] {$e$};
\draw[heavyedge] (v)--(z); \draw[backedge,bend left=34] (z) to (u);
\node[align=center] at (0,-.65) {self-hole};
\node[vnode] (a2) at (4.1,2.8) {$a$}; \node[vnode] (c2) at (4.1,1.9) {}; \node[vnode] (p2) at (4.1,.95) {}; \node[vnode] (z2) at (4.1,0) {$z$};
\draw[lightedge] (a2)--(c2) node[midway,right] {$f\notin \Packing$}; \draw[heavyedge] (c2)--(p2);
\draw[lightedge] (p2)--(z2); \draw[packingedge] (p2)--(z2) node[midway,right=3pt] {$e$};
\draw[backedge,bend left=34] (z2) to (a2);
\node[align=center] at (4.1,-.65) {foreign hole};
\node[vnode] (a3) at (8.2,2.8) {$a$}; \node[vnode] (c3) at (8.2,1.9) {}; \node[vnode] (p3) at (8.2,.95) {}; \node[vnode] (z3) at (8.2,0) {$z$};
\draw[gateedge] (a3)--(c3) node[midway,right] {$g$}; \draw[heavyedge] (c3)--(p3);
\draw[lightedge] (p3)--(z3); \draw[packingedge] (p3)--(z3) node[midway,right=3pt] {$e$};
\draw[backedge,bend left=34] (z3) to (a3);
\node[align=center] at (8.2,-.65) {gated copy};
\end{tikzpicture}
\caption{The three channels for an edge in a packed matching bundle.  The blue overlay marks the packed matching edge $e$ on the $a$--$z$ tree path; in the foreign-hole and gated-copy panels it lies below the first edge $f$ or $g$, with any intervening tree edges compressed.}
\label{fig:channels}
\end{figure}
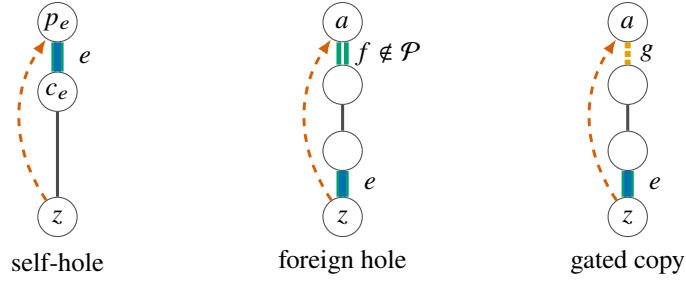

\subsection{Foreign holes}

\begin{lemma}[foreign-hole budget]\label{lem:foreign}
The total foreign-hole mass is at most $\Missing+2\LPex$.
\end{lemma}
\begin{proof}
For a foreign hole, write its first matching edge as $f=p_f c_f$, where $p_f$ is the parent of $c_f$.  The uplink has top endpoint $p_f$.  All such copies are heavy edges incident to $p_f$, whose total heavy degree is $1+\vex{p_f}$.  The $\Missing$ matching edges missing from the packing have distinct upper endpoints.  Summing and using \eqref{eq:eps-budget} proves the claim.
\end{proof}

\subsection{Gated copies}

\begin{lemma}[gate domination]\label{lem:gate-domination}
If a gated copy $\ell$ has top endpoint $a$ and branch gate $g$, then
\begin{equation}\label{eq:domination}
 x^*_g\ge x^*_\ell.
\end{equation}
Moreover, $g\notin\Packing$.
\end{lemma}
\begin{proof}
When DFS at $a$ first enters the child branch containing the source $z$ of $\ell$ through the tree edge $g$, the vertex $z$ is still unvisited.  Hence the heavy edge $\ell=az$ is an eligible candidate at that iteration.  The priority rule selects $g$ with maximum $x^*$-value among the eligible heavy edges, proving \eqref{eq:domination}.  Finally, $\ell$ covers both $g$ and its packed matching edge, so the packing property implies $g\notin\Packing$.
\end{proof}

For a top endpoint $a$, let $C_a$ be the total mass of gated copies with top endpoint $a$, and 
let $G_a$ be the total $x^*$-mass of the distinct heavy branch gates used by gated copies with top endpoint $a$.
Distinct top-endpoint--branch pairs give distinct gates, all in $\HeavyTree\setminus \PackH$, so
\begin{equation}\label{eq:G-total}
 \sum_a G_a\le \OutsideMass.
\end{equation}
Define
\[
 O_a:=(C_a-2G_a)_{+},
 \qquad
 \GateOver:=\sum_a O_a.
\]
By \eqref{eq:G-total},
\begin{equation}\label{eq:gate-copy-total}
 \sum_a C_a\le2\OutsideMass+\GateOver.
\end{equation}

The next elementary inequality bounds local overload using the gate residual and singleton excess.

\begin{lemma}[local deficit inequality]\label{lem:local}
Let $k\ge1$ and $y_1,\ldots,y_k\in[0,1]$, and put
\[
 C:=\sum_{i=1}^{k}y_i,
 \qquad
 R:=\sum_{i=1}^{k}(1-y_i)=k-C,
\]
and let $G,H\ge0$.  If
\[
 y_i\le G\quad(i=1,\ldots,k),
 \qquad
 C+G\le1+H,
\]
then
\[
 (C-2G)_{+}\le\frac18R+H.
\]
\end{lemma}
\begin{proof}
If $C\le2G$, the left side is zero.  Assume $C>2G$.  Since $C\le kG$, we have $G>0$ and $k\ge3$.  After multiplying by eight and rearranging, it suffices to prove
\begin{equation}\label{eq:local-rearranged}
 9C-16G-8H\le k.
\end{equation}
The hypothesis gives $H\ge(C+G-1)_{+}$.

If $C+G\le1$, then $C\le\min\{kG,1-G\}$.  The piecewise-linear expression $9\min\{kG,1-G\}-16G$ is maximized at $G=1/(k+1)$, and hence the left side of \eqref{eq:local-rearranged} is at most
\[
 \frac{9k-16}{k+1}\le k,
\]
where the last inequality is equivalent to $(k-4)^2\ge0$.

Suppose $C+G>1$.  Then $H\ge C+G-1$, and the left side of \eqref{eq:local-rearranged} is at most $C-24G+8$.  If $G\le1$, then $C\le kG$.  For $k\le24$, the inequality $C+G>1$ implies $G>1/(k+1)$, and
\[
 C-24G+8\le\frac{9k-16}{k+1}\le k.
\]
For $k\ge24$, it is at most $(k-24)G+8\le k-16$.  Finally, if $G\ge1$, then $C\le k$ and $C-24G+8\le k-16$.  This proves \eqref{eq:local-rearranged}.
\end{proof}

At a top endpoint $a$, apply the lemma to the values of the gated copies, with $G=G_a$ and
\[
 H := \one_{\{a\text{ unmatched}\}}+\vex{a}.
\]
The domination condition follows from \cref{lem:gate-domination}, and the heavy-degree budget gives $C_a+G_a\le1+H$.

Let $\GateCopies$ be all gated copies and define their incidence residual
\[
 \GateResidual:=\sum_{e\in\GateCopies}(1-x^*_e).
\]
Summing \cref{lem:local}, using \cref{lem:support-volume} and \eqref{eq:eps-budget}, gives
\begin{equation}\label{eq:gate-over-first}
 \GateOver\le\frac18(m+\TreeMass)+(u+2\LPex).
\end{equation}
Using \eqref{eq:tau-upper} in \eqref{eq:gate-over-first},
\begin{equation}\label{eq:gate-over-final}
 \GateOver
 \le\frac18m+\frac1{16}(\Missing+\OutsideMass)+\frac{17}{16}u+\frac{33}{16}\LPex.
\end{equation}

\subsection{The regional ledger}

Every packed matching bundle has mass at least one.  Summing its self-hole, foreign-hole, and gated-copy parts, then using \cref{lem:foreign} and \eqref{eq:gate-copy-total}, gives
\[
 m-\Missing=|\PackM|
 \le \SelfMass+(\Missing+2\LPex)+(2\OutsideMass+\GateOver).
\]
Therefore
\begin{equation}\label{eq:regional-ledger}
 m\le2\Missing+2\OutsideMass+\SelfMass+2\LPex+\GateOver.
\end{equation}
Substituting \eqref{eq:gate-over-final} into \eqref{eq:regional-ledger} and multiplying by sixteen yields
\begin{equation}\label{eq:14m-pre}
 14m\le33(\Missing+\OutsideMass)+16\SelfMass+17u+65\LPex.
\end{equation}

\section{The self-hole cut identity}\label{sec:selfhole}

Self-holes can occur in nested DFS subtrees, so a direct disjoint charging is unavailable.  A two-cut identity shows that their mass is already paid by the exact slack in \cref{sec:slack}.

Fix $e=p_e c_e\in \PackM$, with $p_e$ the parent of $c_e$, and let $T_{c_e}$ be the DFS subtree rooted at $c_e$.  Recall
\[
 \localPackSlack{e}=x^*(\Uplinks(e))-1.
\]
The fundamental cut of $e$ has value
\begin{equation}\label{eq:subtree-cut}
 x^*(\delta(T_{c_e}))=1+x^*(\Uplinks(e))=2+\localPackSlack{e}.
\end{equation}
The singleton cut at $p_e$ has value
\begin{equation}\label{eq:upper-singleton}
 x^*(\delta(\{p_e\}))=2+\vex{p_e}.
\end{equation}
\Cref{fig:selfhole} illustrates the two cuts and the internal mass appearing in this identity.
The positive-support LP mass between $\{p_e\}$ and $T_{c_e}$ consists of the matching tree edge $e$ and the self-hole copies of $\Uplinks(e)$, and therefore equals $1+\localSelfMass{e}$.

Apply the identity
\[
 x(\delta A)+x(\delta B)
 =x(\delta(A\cup B))+2x(E(A,B))
\]
to the disjoint sets $A=\{p_e\}$ and $B=T_{c_e}$.  Using \eqref{eq:subtree-cut} and \eqref{eq:upper-singleton},
\begin{equation}\label{eq:selfhole-identity}
 x^*\!\left(\delta(\{p_e\}\cup T_{c_e})\right)
 =2+\vex{p_e}+\localPackSlack{e}-2\localSelfMass{e}.
\end{equation}

\begin{figure}[t]
\centering
\begin{tikzpicture}[scale=.9,every node/.style={font=\small}]
\node[draw=figblue,rounded corners,fill=figlightblue,minimum width=30mm,minimum height=20mm] (Tv) at (3,0) {$T_{c_e}$};
\node[vnode] (u) at (0,0) {$p_e$};
\node[vnode] (v) at (2.1,0) {$c_e$};
\draw[lightedge] (u)--(v) node[midway,above] {$1$};
\draw[backedge,bend left=28] (Tv.south west) to node[below] {$\localSelfMass{e}$} (u.south);
\draw[decorate,decoration={brace,amplitude=5pt,mirror}] (-.5,-1.5)--(4.5,-1.5) node[midway,below=6pt] {$\{p_e\}\cup T_{c_e}$};
\node[align=left,anchor=west] at (5.1,.15) {$x(\delta(\{p_e\}))=2+\vex{p_e}$\\$x(\delta(T_{c_e}))=2+\localPackSlack{e}$};
\end{tikzpicture}
\caption{The self-hole identity.  The positive-support internal mass between the singleton and the subtree is $1+\localSelfMass{e}$, so every proper self-hole is paid by singleton excess and local packing slack.}
\label{fig:selfhole}
\end{figure}
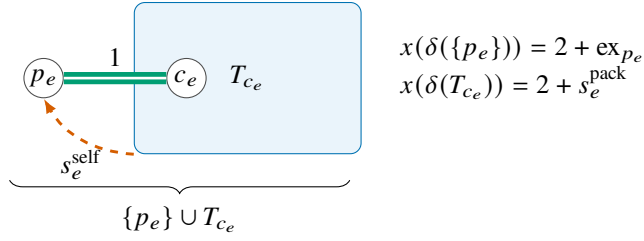

\begin{lemma}[self-hole theorem]\label{lem:selfhole}
\[
 \SelfMass\le1+\LPex+\frac{\PackSlack}{2}.
\]
\end{lemma}
\begin{proof}
If $\{p_e\}\cup T_{c_e}$ is a proper subset of $V$, cut feasibility and \eqref{eq:selfhole-identity} give
\[
 2\localSelfMass{e}\le\vex{p_e}+\localPackSlack{e}.
\]
The union can equal $V$ only when $p_e$ is the DFS root and $c_e$ is its unique child.  At most one packed matching edge has this property.  For that exceptional edge, the left side of \eqref{eq:selfhole-identity} is zero, and hence
\[
 \localSelfMass{e}=1+\frac12(\vex{p_e}+\localPackSlack{e}).
\]
Combining the proper-cut and exceptional cases, for every $e\in\PackM$,
\[
 \localSelfMass{e}\le \one_{\{e\text{ exceptional}\}}
 +\frac12(\vex{p_e}+\localPackSlack{e}).
\]
The upper endpoints $p_e$ of packed matching edges are distinct, so
\[
 \sum_{e\in \PackM}\vex{p_e}\le2\LPex.
\]
By the slack decomposition \eqref{eq:sigma-decomp},
\[
 \sum_{e\in \PackM}\localPackSlack{e}\le\PackSlack.
\]
Summing completes the proof.
\end{proof}

\begin{remark}[tight root exception]\label{rem:root-exception}
The additive one in \cref{lem:selfhole} is necessary for this argument.  In the envelope execution of \cref{subsec:envelope}, the packed root edge $ab$ is exceptional and has self-hole mass $\SelfMass=1$, while $\LPex=\PackSlack=0$.
\end{remark}

\section{Proof of the main theorem}\label{sec:main-proof}

Following the outline in \cref{sec:intro}, insert \cref{lem:selfhole} into \eqref{eq:14m-pre}:
\[
\begin{aligned}
 14m
 &\le33(\Missing+\OutsideMass)+16\left(1+\LPex+\frac\PackSlack2\right)+17u+65\LPex\\
 &=33(\Missing+\OutsideMass)+8\PackSlack+17u+81\LPex+16.
\end{aligned}
\]
Since $\Saving=\Missing+\OutsideMass+\PackSlack$ and $\PackSlack\ge0$,
\[
 33(\Missing+\OutsideMass)+8\PackSlack=33\Saving-25\PackSlack\le33\Saving.
\]
Therefore
\begin{equation}\label{eq:saving-lower}
 \Saving\ge\frac{14m-17u-81\LPex-16}{33}.
\end{equation}
Substitute \eqref{eq:saving-lower} into the exact cost identity \eqref{eq:alg-exact}:
\[
\begin{aligned}
 \ALG
 &\le2m+\frac32u+\frac12\LPex-1
 -\frac{14m-17u-81\LPex-16}{66}\\
 &=\frac{59m+58u+57\LPex-25}{33}
 \le\frac{59}{33}(m+u+\LPex)-\frac{25}{33}\\
 &=\frac{59}{33}c^*-\frac{25}{33}.
\end{aligned}
\]
This proves \cref{thm:main}.

\section{Forest augmentation at the minimum LP value}\label{sec:fap-face}

The \emph{Forest Augmentation Problem} (FAP) is the analogous $0/1$-cost problem in which the zero-cost edges form an arbitrary forest.  The \emph{Path Augmentation Problem} (PAP) is the special case in which every component of that forest is a path.  \citet{BDS} observed that their DFS approach extends to PAP when the cut LP value equals the number of free components.  We record the consequence of \cref{thm:main} in a form that also shows that, in this minimum-value regime, every FAP instance is necessarily of this type.

Let $G=(V,F_0\mathbin{\dot\cup}H)$ be a finite loopless FAP multigraph, where $F_0$ is the zero-cost forest and every edge in $H$ has unit cost.  Let $k$ be the number of components of $F_0$.  The standard cut relaxation is \eqref{lp:map} with $M$ replaced by $F_0$ and objective $c(x)=x(H)$.  Write $\operatorname{OPT}_{\mathrm{LP}}(G,F_0)$ for its optimum.  Since free coordinates have zero objective coefficient, they may be normalized to one.

\begin{lemma}[minimum-value structure]\label{lem:fap-face}
Suppose $|V|\ge2$, and let $x$ be a feasible normalized FAP solution, so $x_f=1$ for every $f\in F_0$.  Then
\[
 c(x)\ge k.
\]
If equality holds, every component of $F_0$ is a path (possibly an isolated vertex), and no positive-support heavy edge is incident with an internal vertex of one of these paths.
\end{lemma}
\begin{proof}
Write $n=|V|$.  Since $|F_0|=n-k$, summing the singleton constraints gives the exact excess identity
\begin{equation}\label{eq:fap-excess}
 \sum_{v\in V}\bigl(x(\delta(v))-2\bigr)
 =2\bigl(|F_0|+c(x)-n\bigr)
 =2\bigl(c(x)-k\bigr).
\end{equation}
Every summand on the left is nonnegative, proving $c(x)\ge k$.  If $c(x)=k$, every singleton constraint is tight.  Hence, for every vertex $v$,
\[
 d_{F_0}(v)+x(\delta(v)\cap H)=2.
\]
It follows that $d_{F_0}(v)\le2$, so every forest component is a path.  If $v$ is an internal path vertex, then $d_{F_0}(v)=2$ and therefore $x(\delta(v)\cap H)=0$.  Nonnegativity of the coordinates implies that every incident heavy edge has value zero.
\end{proof}

Equation~\eqref{eq:fap-excess} also gives the following stability bound for every normalized feasible point on an instance with at least two vertices:
\begin{equation}\label{eq:fap-stability}
 \sum_{v\in V}\bigl(d_{F_0}(v)-2\bigr)_{+}
 +\sum_{v:\,d_{F_0}(v)=2}x(\delta(v)\cap H)
 \le 2\bigl(c(x)-k\bigr).
\end{equation}
Indeed, the first summand at a vertex of free degree greater than two and the second summand at a vertex of free degree two are each bounded by $x(\delta(v))-2$; the relevant vertex sets are disjoint.

\begin{corollary}[FAP at the minimum LP value]\label{cor:fap-face}
Suppose that the optimum value of the standard cut LP for a FAP instance is $k$, the number of components of its zero-cost forest.  Then, in polynomial time, one can find a 2-edge-connected spanning submultigraph using at most
\[
 \frac{59}{33}k-\frac{25}{33}
 =\frac{59}{33}\operatorname{OPT}_{\mathrm{LP}}-\frac{25}{33}
\]
unit-cost edges.  Consequently, the integrality gap of the cut LP over instances in this minimum-value regime is at most $59/33\approx1.788$.
\end{corollary}
\begin{proof}
Choose an optimum LP solution $x$ and normalize every free coordinate to one.  The hypothesis excludes $|V|=1$: in that case $k=1$, while the absence of nonempty proper cuts gives LP optimum zero.  Hence \cref{lem:fap-face} applies, and every component of $F_0$ is a path and positive-support heavy edges are incident only with path endpoints or isolated free vertices.  Delete the zero-valued heavy edges.  For every nontrivial free path, suppress its internal vertices and replace the path by one zero-cost edge joining its endpoints; retain every isolated free vertex.  The resulting zero-cost edges form a matching $M'$ in a loopless multigraph $G'$.  The compressed graph $G'$ has at least two vertices: either a nontrivial free path leaves its two endpoints, or $F_0$ has at least two isolated components.

The heavy coordinates of $x$, together with value one on $M'$, define a feasible solution $x'$ of the MAP cut LP on $G'$.  To verify this, fix a nonempty proper subset $S'\subsetneq V(G')$.  Construct a set $S\subsetneq V$ path by path.  If the two endpoints lie on the same side of $S'$, put the whole path on that side.  If they lie on opposite sides, split the path at one edge, putting the two resulting subpaths on the prescribed sides.  No positive-support heavy edge is incident with a suppressed vertex, so the heavy contribution to the two cuts is identical; a compressed matching edge crosses $S'$ exactly when one original free-path edge crosses $S$.  Therefore
\(
 x'(\delta_{G'}(S'))=x(\delta_G(S))\ge2.
\)
In particular, $c(x')=k$.

Let $m=|M'|$ and let $u$ be the number of unmatched vertices of $G'$.  Each nontrivial free path contributes one matching edge and each isolated free vertex contributes one unmatched vertex, so
\(
 m+u=k.
\)

Conversely, every feasible point of the MAP cut LP on $G'$ can be normalized to value one on $M'$, and summing its singleton constraints gives heavy cost at least $m+u=k$.  Hence the MAP cut LP optimum on $G'$ is exactly $k$.

Apply \cref{thm:main} to $G'$.  It returns a 2-edge-connected spanning submultigraph using at most $59k/33-25/33$ heavy edges.  Finally, replace every matching edge of this solution by its original free path.  Subdividing an edge of a bridgeless connected multigraph preserves bridgelessness and connectivity, so the expanded submultigraph is a feasible FAP solution of the same unit cost.  All steps are polynomial-time.
\end{proof}

Thus the minimum-value PAP consequence noted by \citet{BDS} inherits the full $59/33\approx1.788$ bound.  No separate path-forest assumption is needed: \cref{lem:fap-face} derives the path structure directly from equality in the elementary LP lower bound.

\section{Discussion}\label{sec:discussion}

The proof illustrates three principles that may be useful beyond MAP.

\paragraph{Exploit the exact packing certificate of an exactly solved residual problem.}
The augmentation is solved optimally, so its packing certificate contains more information than the feasible residual LP vector alone.  Retaining the packing slack turns the approximation analysis into an exact cost identity.

\paragraph{Count fractional support relative to what is already integral.}
The fractional support of an extreme cut LP point can be large in absolute terms, but the unit-valued skeleton severely restricts which tight cuts can contribute independent rank.  The minimum-cut dimension theorem (\cref{thm:mincut-dimension}) and the one-bridge quotient argument convert this structure into a support-volume bound.

\paragraph{Search for a cut identity before building a stack analysis.}
Self-holes initially appear to form nested DFS obstructions.  The identity \eqref{eq:selfhole-identity} shows that they are almost completely local: their mass is already paid by singleton excess and packed-cut slack.

The envelope in \cref{subsec:envelope} is also the classical $4/3$ integrality-gap example.  Thus, over all feasible instances, roots, and deterministic tie-breaking orders, the supremum of the LP-relative ratio of this DFS pipeline lies between $4/3\approx1.333$ and $59/33\approx1.788$.  We do not know whether the analysis is tight.  Closing this gap---either by a sharper analysis of the same three-step algorithm or by a different LP-relative algorithm---is a concrete open problem.  Approaching the established integral-optimum frontier of $13/8=1.625$ will likely require additional structure beyond the present regional accounting.

The minimum-value reduction in \cref{sec:fap-face} also isolates the exact point at which FAP departs from MAP: once the LP value exceeds the number of free components, the free forest may branch and positive heavy support may meet internal path vertices.  Equations~\eqref{eq:fap-excess} and \eqref{eq:fap-stability} show that both effects are globally paid by the LP excess $c(x)-k$; converting that budget into an explicit approximation guarantee remains open.
\section*{Acknowledgment}
The authors would like to thank Waldo Gálvez, Ermiya Farokhnejad, and Atrayee Majumder for fruitful discussions.

\section*{Declarations}
\paragraph{Use of generative AI.}

Generative AI tools were used in formulating theorems and proofs, language editing, notation review, literature-search, 
citation preparation, figure production, and preparation of the \LaTeX{} source.  
The authors checked the mathematical content and citations.  
The authors assume responsibility for all content.

\appendix
\section{Technical completeness and boundary cases}\label{app:technical}

For completeness, we state the remaining multigraph conventions and boundary cases.

\paragraph{Support restriction.}
The algorithm augments the DFS tree using only positive non-tree support edges.  By \eqref{eq:fundamental}, these edges fractionally cover every tree edge and hence form a feasible uplink instance.  If one instead allows all original edges, the optimum augmentation can only become cheaper, so the same approximation bound holds.

\paragraph{Parallel copies.}
Parallel copies are separate LP coordinates and separate edge copies in every cut.  If two heavy copies with the same endpoints were both fractional at an extreme point, increasing one and decreasing the other by a sufficiently small common amount would preserve all cuts, the objective, and the bounds.  Thus at most one copy in a parallel class is fractional.  The rank proof does not need this observation, but it confirms that merging parallel fractional columns in the type-zero minimum-cut argument loses no rank.

\paragraph{Gate distinctness.}
A gate is identified by its top endpoint and child branch.  Distinct top--branch pairs give distinct tree edges.  Every gated copy covers both its gate and its packed matching edge, so no gate belongs to $\Packing$.  Hence the total gate mass is bounded by $\OutsideMass=x^*(\HeavyTree\setminus \PackH)$.

\paragraph{Zero-valued edges.}
All bundle identities concern positive support edges.  A zero-valued edge may cross a displayed cut, but it contributes zero and does not affect any equality involving $x^*$.

\paragraph{The root exception.}
The set $\{p_e\}\cup T_{c_e}$ equals $V$ only if $p_e$ is the root and $c_e$ is its unique child.  Since $M$ is a matching, at most one packed matching edge can be exceptional.

\paragraph{The cases $m=0$ and $|V|=2$.}
If $m=0$, the packed-light regional analysis is empty; all displayed inequalities remain valid.  A feasible two-vertex instance is solved optimally, and the additive statement of \cref{thm:main} is valid directly.

% Inline bibliography, in the original alphabetic order.
\begingroup
\setlength{\labelsep}{10pt}

\endgroup

\end{document}